\documentclass[conference]{IEEEtran}
\usepackage{mathrsfs}
\usepackage[square,numbers,sort&compress]{natbib}

\usepackage{ifpdf}

\usepackage{amsmath}
\usepackage{amssymb}
\usepackage{enumerate}
\usepackage{color,xcolor}
\usepackage{graphicx}
\usepackage{subfigure}
\usepackage{pifont}

\ifCLASSINFOpdf
\else
\fi

\newtheorem{definition}{Definition}
\newtheorem{theorem}{Theorem}

\newtheorem{lemma}{Lemma}

\begin{document}
%
\title{An Information-Spectrum Approach to the Capacity Region of General Interference Channel}



\author{\authorblockN{Lei Lin\authorrefmark{1}, Xiao Ma\authorrefmark{2}, Xiujie Huang\authorrefmark{2} and Baoming Bai\authorrefmark{3}
\authorblockA{\authorrefmark{1}Department of Applied Mathematics, Sun Yat-sen University, Guangzhou 510275, Guangdong, China}
\authorblockA{\authorrefmark{2}Department of Electronics and Communication
Engineering, Sun Yat-sen University, Guangzhou 510006, GD, China}
\authorblockA{\authorrefmark{3}State Lab. of ISN, Xidian University, Xi'an 710071, Shaanxi, China}}
Email: linlei2@mail2.sysu.edu.cn, maxiao@mail.sysu.edu.cn}


%


\maketitle

\begin{abstract}
This paper is concerned with general interference channels characterized by a sequence of transition~(conditional) probabilities.
We present a general formula for the capacity region of the interference channel with two pairs of users. The formula shows that the capacity region is the union of a family of rectangles, where each rectangle is determined by a pair of spectral inf-mutual information rates. Although the presented formula is usually difficult to compute, it provides us useful insights into the interference channels. For example, the formula suggests us that the simplest inner bounds~(obtained by treating the interference as noise) could be improved by taking into account the structure of the interference processes. This is verified numerically by computing the mutual information rates for Gaussian interference channels with embedded convolutional codes.

\end{abstract}




%
\IEEEpeerreviewmaketitle

\section{Introduction}
The interference channel~(IC) is a communication model with multiple pairs of senders and receivers, in which each sender has an independent message intended only for the corresponding receiver. This model was first mentioned by Shannon~\cite{Shannon61} in 1961 and further studied by Ahlswede~\cite{Ahlswede74} in 1974. A basic problem for the IC is to determine the capacity region, which is currently one of long-standing open problems in information theory. Only in some special cases, the capacity regions are known, such as strong interference channels and deterministic interference channels~\cite{Carleial75,Han81,ElGamal82,Costa87}. On the other hand, some inner and outer bounds of the capacity region are obtained, for example, see~\cite{Han81,Kramer04,Etkin08}. In recent years, Etkin, Tse and Wang~\cite{Etkin08} introduced the idea of approximation to show that Han and Kobayashi region~(HK region)~\cite{Han81} is within one bit of the capacity region for Gaussian interference channel~(GIFC).

In~\cite{XiaoMa2011}, the authors proposed a new computational model for the two-user GIFC, in which one pair of users~(called {\em primary users}) are constrained to use a fixed encoder and the other pair of users~(called {\em secondary users}) are allowed to optimize their codes. The maximum rate at which the secondary users can communicate reliably without degrading the performance of the primary users is called the {\em accessible capacity} of the secondary users. Since the structure of the interference from the primary link has been taken into account in the computation, the accessible capacity is usually higher than the maximum rate when treating the interference as noise, as is consistent with the spirit of~\cite{FrancoisArticle2011}\cite{Moshksar11}. However, to compute the accessible capacity~\cite{XiaoMa2011}, the primary link is allowed to have a non-neglected error probability. This makes the model unattractive when the capacity region is considered. For this reason, we will remove the fixed-code constraints on the primary users in this paper.\footnote{The authors are grateful to Prof. Tse who inspired us to continue the research in this direction when our previous work was presented at ISIT'2011.} In other words, we will compute a pair of transmission rates at which both links can be asymptotically error-free.

To justify the computational results, we consider a more general interference channel which is characterized by a sequence of transition probabilities ${\bf W} = \{W^n\}_{n=1}^{\infty}$. Similar to~\cite{Han98}, we utilize the information spectrum approach~\cite{Han93}\cite{Han03}. The capacity region can be described in terms of the {\em spectral inf-mutual information rates}.

The rest of the paper is structured as follows. Sec.~\ref{Definition_section} introduces the main definitions, including general IC
and spectral inf-mutual information rate. In Sec.~\ref{Capacity_section}, a formula is proved for the capacity region of the general IC.
In Sec.~\ref{GIFC_section} we present numerical results for GIFC with binary-phase shift-keying~(BPSK) modulation. Sec.~\ref{Conclusions} concludes this paper.
\begin{figure}
  \includegraphics[width=8.5cm]{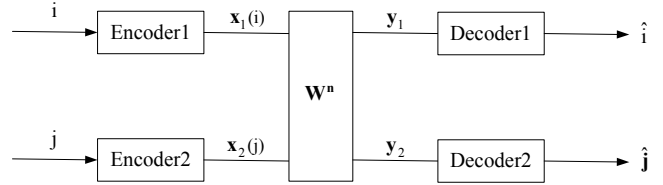}\\
  \caption{General interference channel ${\bf W}$.}\label{general IFC_Fig}
\end{figure}

\section{Basic Definitions And Problem Statement}\label{Definition_section}
\subsection{General IC}
Let $\mathcal{X}_1$, $\mathcal{X}_2$ be two finite input alphabets and $\mathcal{Y}_1$, $\mathcal{Y}_2$ be two finite output alphabets. A general interference channel ${\bf W}$~(see Fig.~\ref{general IFC_Fig}) is characterized by a sequence ${\bf W} = \{W^n(\cdot,\cdot|\cdot,\cdot)\}_{n=1}^{\infty}$, where $W^n: \mathcal{X}_1^n \times \mathcal{X}_2^n \rightarrow \mathcal{Y}_1^n \times \mathcal{Y}_2^n$ is a probability transition matrix. That is, for all $n$,
\begin{eqnarray*}
  W^n({\bf y}_1,{\bf y}_2|{\bf x}_1,{\bf x}_2)  &\geq&  0 \\
  \sum\limits_{{\bf y}_1 \in \mathcal{Y}_1^n,{\bf y}_2 \in \mathcal{Y}_2^n} W^n({\bf y}_1,{\bf y}_2|{\bf x}_1,{\bf x}_2) &=& 1.
\end{eqnarray*}
The marginal distributions $W_1^n,W_2^n$ of the $W^n$ are given by
\begin{eqnarray}
W_1^n({\bf y}_1|{\bf x}_1,{\bf x}_2) &=& \sum_{{\bf y}_2 \in \mathcal{Y}_2^n} W^n({\bf y}_1,{\bf y}_2|{\bf x}_1,{\bf x}_2),\\
W_2^n({\bf y}_2|{\bf x}_1,{\bf x}_2) &=& \sum_{{\bf y}_1 \in \mathcal{Y}_1^n} W^n({\bf y}_1,{\bf y}_2|{\bf x}_1,{\bf x}_2).
\end{eqnarray}
\begin{definition}
An $(n,M_n^{(1)},M_n^{(2)},\varepsilon_n^{(1)}, \varepsilon_n^{(2)})$ code for the interference channel ${\bf W}$ consists of the following essentials,
\begin{itemize}
  \item [a)]message sets:
\begin{eqnarray}
\mathcal{M}_n^{(1)} = \{1,2,\ldots,M_n^{(1)}\},\,\,\,&{\rm for\,\,\,sender~1}\nonumber\\
\mathcal{M}_n^{(2)} = \{1,2,\ldots,M_n^{(2)}\},\,\,\,&{\rm for\,\,\,sender~2}\nonumber
\end{eqnarray}

  \item [b)]sets of codewords:
\begin{equation}\nonumber
\begin{array}{ll}
\{{\bf x}_1(1),{\bf x}_1(2),\ldots,{\bf x}_1(M_n^{(1)})\} \subset \mathcal{X}_1^n,\,\,\,&{\rm for\,\,\,encoder~1}\\
\{{\bf x}_2(1),{\bf x}_2(2),\ldots,{\bf x}_2(M_n^{(2)})\} \subset \mathcal{X}_2^n,\,\,\,&{\rm for\,\,\,encoder~2}
\end{array}
\end{equation}

For sender~1 to transmit message $i$, encoder~1 outputs the codeword ${\bf x}_1(i)$. Similarly, for sender~2 to transmit message $i$, encoder~2 outputs the codeword ${\bf x}_2(j)$.

  \item [c)]collections of decoding sets:
\begin{equation}\nonumber
\begin{array}{ll}
\{\mathcal{B}_{1i} \subseteq \mathcal{Y}_1^n\}_{i = 1,...,M_n^{(1)}},\,\,\,&{\rm for\,\,\,decoder~1}\\
\{\mathcal{B}_{2j} \subseteq \mathcal{Y}_2^n\}_{j = 1,...,M_n^{(2)}},\,\,\,&{\rm for\,\,\,decoder~2}
\end{array}
\end{equation}
where $\mathcal{B}_{1i} \bigcap \mathcal{B}_{1i'} = \emptyset$ for $i \neq i'$ and $\mathcal{B}_{2j} \bigcap \mathcal{B}_{2j'} = \emptyset$ for $j \neq j'$.

After receiving ${\bf y}_1$, decoder~1 outputs $\hat i$ whenever ${\bf y}_1 \in \mathcal{B}_{1\hat{i}}$. Similarly, after receiving ${\bf y}_2$, decoder~2 outputs $\hat j$ whenever ${\bf y}_2 \in \mathcal{B}_{2\hat{j}}$.

  \item [d)]probabilities of decoding errors:
\begin{equation}\nonumber
\begin{array}{ll}
      & \varepsilon_n^{(1)} = \frac{1}{M_n^{(1)}M_n^{(2)}}\sum\limits_{i=1}^{M_n^{(1)}} \sum\limits_{j=1}^{M_n^{(2)}}W_1^n(\mathcal{B}^c_{1i}|{\bf x}_1(i), {\bf x}_2(j))\\
      & \varepsilon_n^{(2)} = \frac{1}{M_n^{(1)}M_n^{(2)}}\sum\limits_{i=1}^{M_n^{(1)}} \sum\limits_{j=1}^{M_n^{(2)}}W_2^n(\mathcal{B}^c_{2j}|{\bf x}_1(i), {\bf x}_2(j)),
\end{array}
\end{equation}
where the superscript $``c"$ denotes the complement of a set. Here we have assumed that each message of $i \in \mathcal{M}_n^{(1)}$ and $j \in \mathcal{M}_n^{(2)}$ is produced with equal probability.
\end{itemize}
\end{definition}
{\bf Remark}: The optimal decoding to minimize the probability of errors is defining the decoding sets $\mathcal{B}_{1i}$ and $\mathcal{B}_{2j}$ according to the the maximum likelihood decoding rule~\cite{Etkin10}. That is, the two receivers choose, respectively,
$$\hat{i} = \arg\max_i {\rm Pr}\{{\bf y}_1|{\bf x}_1(i)\}$$
and
$$\hat{j} = \arg\max_j {\rm Pr}\{{\bf y}_2|{\bf x}_2(j)\}$$
as the transmitted messages.

\begin{definition}
A rate pair $(R_1,R_2)$ is achievable if there exists a sequence of $(n,M_n^{(1)},M_n^{(2)},\varepsilon_n^{(1)}, \varepsilon_n^{(2)})$ codes such that
\begin{eqnarray}
\lim_{n \rightarrow \infty} \varepsilon_n^{(1)} = 0 &{\rm and}& \lim_{n \rightarrow \infty} \varepsilon_n^{(2)} = 0 ,\nonumber\\
\liminf_{n \rightarrow \infty} \frac{\log M_n^{(1)}}{n} \geq R_1 &{\rm and}& \liminf_{n \rightarrow \infty} \frac{\log M_n^{(2)}}{n} \geq R_2.\nonumber
\end{eqnarray}
\end{definition}

\begin{definition}
The set of all achievable rates is called the capacity region of the interference channel ${\bf W}$, which is denoted by $\mathcal{C}({\bf W})$.
\end{definition}

\subsection{Preliminaries of Information-Spectrum Approach}
The following notions can be found in~\cite{Han03}.
\begin{definition}[liminf in probability]
For a sequence of random variables $\{Z_n\}_{n=1}^{\infty}$,
$$p\textrm{-}\liminf_{n \rightarrow \infty}Z_n \equiv \sup\{\beta|\lim_{n \rightarrow \infty}{\rm Pr}\{Z_n < \beta\} = 0\}.$$
\end{definition}

\begin{definition}
If two random variables sequences ${\bf X}_1 = \{{X}_1^n\}_{n=1}^{\infty}$ and ${\bf X}_2 = \{{X}_2^n\}_{n=1}^{\infty}$ satisfy that
\begin{equation}\label{P_independ}
P_{{X}_1^n{X}_2^n}({\bf x}_1,{\bf x}_2) = P_{{X}_1^n}({\bf x}_1) P_{{X}_2^n}
({\bf x}_2)
\end{equation}
for all ${\bf x}_1 \in \mathcal{X}_1^n$, ${\bf x}_2 \in \mathcal{X}_2^n$ and $n$, they are called independent and denoted by ${\bf X}_1 \bot {\bf X}_2$.
\end{definition}

Similar to~\cite{Han93}, we have
\begin{definition}
Let $S_I \stackrel{\triangle}{=} \{({\bf X}_1,{\bf X}_2)| {\bf X}_1 \bot {\bf X}_2\}$. Given an $({\bf X}_1,{\bf X}_2) \in S_I$, for the interference channel ${\bf W}$, we define the {\em spectral inf-mutual information rate} by
\begin{eqnarray}
\underline{I}({\bf X}_1;{\bf Y}_1) &\equiv& p\textrm{-}\liminf_{n \rightarrow \infty}\frac{1}{n} \log \frac{P_{Y_1^n|X_1^n}({Y}_1^n | {X}_1^n)}{P_{{Y}_1^n}({Y}_1^n)}\\
\underline{I}({\bf X}_2;{\bf Y}_2) &\equiv& p\textrm{-}\liminf_{n \rightarrow \infty}\frac{1}{n} \log \frac{P_{Y_2^n|X_2^n}({Y}_2^n | {X}_2^n)}{P_{{Y}_2^n}({Y}_2^n)},
\end{eqnarray}
where
\begin{eqnarray}
P_{Y_1^n|X_1^n}({\bf y}_1|{\bf x}_1) &=& \sum_{{\bf x}_2 ,{\bf y}_2} P_{X_2^n}({\bf x}_2)W^n({\bf y}_1,{\bf y}_2|{\bf x}_1,{\bf x}_2),\label{Py1x1}\\
P_{Y_2^n|X_2^n}({\bf y}_2|{\bf x}_2) &=& \sum_{{\bf x}_1 ,{\bf y}_1} P_{X_1^n}({\bf x}_1)W^n({\bf y}_1,{\bf y}_2|{\bf x}_1,{\bf x}_2).\label{Py2x2}
\end{eqnarray}
\end{definition}

\section{The Capacity Region of General IC}\label{Capacity_section}
In this section, we derive a formula for the capacity region $\mathcal{C}({\bf W})$ of the general IC.

\begin{theorem}\label{Capacity_Theorem}
The capacity region $\mathcal{C}({\bf W})$ of the interference channel ${\bf W}$ is given by
\begin{equation}
\mathcal{C}({\bf W}) = \bigcup_{({\bf X}_1,{\bf X}_2) \in S_I} \mathcal{R}_{\bf W}({\bf X}_1,{\bf X}_2),
\end{equation}
where $\mathcal{R}_{\bf W}({\bf X}_1,{\bf X}_2)$ is defined as the collection of all $(R_1,R_2)$ satisfying that
\begin{eqnarray}
0 \leq R_1 &\leq& \underline{I}({\bf X}_1;{\bf Y}_1)\label{Capacity_R1}\\
0 \leq R_2 &\leq& \underline{I}({\bf X}_2;{\bf Y}_2)\label{Capacity_R2}.
\end{eqnarray}
\end{theorem}

To prove Theorem~\ref{Capacity_Theorem}, we need the following lemmas.

\begin{lemma}\label{Error_lemma_1}
Let
$$({\bf X}_1 = \{{X}_1^n\}_{n=1}^{\infty},{\bf X}_2 = \{{X}_2^n\}_{n=1}^{\infty})$$
be any channel input such that $({\bf X}_1,{\bf X}_2) \in S_I$. The corresponding output via an interference channel ${\bf W} = \{W^n\}$ is denoted by $({\bf Y}_1=\{{Y}_1^n\}_{n=1}^{\infty},{\bf Y}_2 = \{{Y}_2^n\}_{n=1}^{\infty})$. Then, for any fixed $M_n^{(1)}$ and $M_n^{(2)}$, there exists an $(n,M_n^{(1)},M_n^{(2)},\varepsilon_n^{(1)}, \varepsilon_n^{(2)})$ code satisfying that
\begin{equation}\label{Lamma_1_inequa}
\varepsilon_n^{(1)} + \varepsilon_n^{(2)} \leq {\rm Pr}\{T^c_n(1)\} + {\rm Pr}\{T^c_n(2)\} + 2e^{-n\gamma},
\end{equation}
where
\begin{equation}
\begin{array}{l}
T_n(1) = \{({\bf x}_1,{\bf y}_1)| \frac{1}{n} \log \frac{P_{Y_1^n|X_1^n}({\bf y}_1 | {\bf x}_1)}{P_{{Y}_1^n}({\bf y}_1)} > \frac{1}{n} \log M_n^{(1)} + \gamma\}\nonumber\\
T_n(2) = \{({\bf x}_2,{\bf y}_2)| \frac{1}{n} \log \frac{P_{Y_2^n|X_2^n}({\bf y}_2 | {\bf x}_2)}{P_{{Y}_2^n}({\bf y}_2)} > \frac{1}{n} \log M_n^{(2)} + \gamma\}.\nonumber
\end{array}
\end{equation}

\end{lemma}
\begin{proof}[Proof of Lemma~\ref{Error_lemma_1}]
The proof is similar to that of~\cite[Lemma~3]{Han93}.

{\bf Codebook generation}. Generate $M_n^{(1)}$ independent codewords ${\bf x}_1(1),...,{\bf x}_1(M_n^{(1)}) \in \mathcal{X}_1^n$ subject to the probability distribution $P_{X_1^n}$. Similarly, generate $M_n^{(2)}$ independent codewords ${\bf x}_2(1),...,{\bf x}_2(M_n^{(2)}) \in \mathcal{X}_2^n$ subject to the probability distribution $P_{X_2^n}$.

{\bf Encoding.} To send message $i$, sender~1 sends the codeword ${\bf x}_1(i)$. Similarly, to send message $j$, sender~2 sends ${\bf x}_2(j)$.

{\bf Decoding.} The receiver~1 chooses $i$ such that $({\bf x}_1(i), {\bf y}_1) \in T_n(1)$ if such $i$ exists and is unique. Similarly, the receiver~2 chooses the $j$ such that $({\bf x}_2(j), {\bf y}_2) \in T_n(2)$ if such $j$ exists and is unique. Otherwise, an error is declared.

 {\bf Analysis of the error probability.} By the symmetry of the random code construction, we can assume that $(1,1)$ was sent. Define
 $$E_{1i} = \{({\bf x}_1(i), {\bf y}_1) \in T_n(1)\},\,\, E_{2j} = \{({\bf x}_2(j), {\bf y}_2) \in T_n(1)\},$$
For receiver~1, an error occurs if $({\bf x}_1(1), {\bf y}_1) \notin T_n(1)$ or $({\bf x}_1(i), {\bf y}_1) \in T_n(1)$ for some $i \neq 1$.
 Similarly, for receiver~2, an error occurs if $({\bf x}_2(1), {\bf y}_2) \notin T_n(2)$ or $({\bf x}_2(j), {\bf y}_2) \in T_n(2)$ for some $j \neq 1$. So the ensemble average of the error probabilities of decoder~1 and decoder~2 can be upper-bounded as follows:
 \begin{equation}\nonumber
 \begin{array}{l}
 \overline{{\varepsilon}_n^{(1)}+{\varepsilon}_n^{(2)}} = \overline{\varepsilon_n^{(1)}} + \overline{\varepsilon_n^{(2)}}\\
 \leq {\rm Pr}\{E_{11}^c\} + {\rm Pr}\{\bigcup\limits_{i \neq 1} E_{1i}\} + {\rm Pr}\{E_{21}^c\} + {\rm Pr}\{\bigcup\limits_{j \neq 1} E_{2j}\}
 \end{array}
 \end{equation}
It can be seen that
  \begin{equation}\nonumber
 \begin{array}{ll}
 &{\rm Pr}\{\bigcup\limits_{i \neq 1} E_{1i}\} \leq \sum\limits_{i \neq 1}{\rm Pr}\{E_{1i}\}\\
 &=\sum\limits_{i \neq 1} {\rm Pr}\{({\bf x}_1(i), {\bf y}_1) \in T_n(1)\}\\
&\stackrel{a}{=} \sum\limits_{i \neq 1} \sum\limits_{({\bf x}_1,{\bf y}_1) \in T_n(1)}P_{X_1^n}({\bf x}_1)P_{Y_1^n}({\bf y}_1)\\
&\stackrel{b}{\leq} \sum\limits_{i \neq 1} \sum\limits_{({\bf x}_1,{\bf y}_1) \in T_n(1)}P_{X_1^n}({\bf x}_1)P_{Y_1^n|X_1^n}({\bf y}_1|{\bf x}_1)\frac{e^{-n\gamma}}{M_n^{(1)}}\\
&\leq  \sum\limits_{i \neq 1}\frac{e^{-n\gamma}}{M_n^{(1)}} = (M_n^{(1)}- 1)\frac{e^{-n\gamma}}{M_n^{(1)}} \leq e^{-n\gamma},
 \end{array}
 \end{equation}
 where $``a"$ follows from the independence of ${\bf x}_1(i)~(i \neq 1)$ and ${\bf y}_1$ and $``b"$ follows from the definition of $T_n(1)$. Similarly, we obtain
 \begin{equation}
 {\rm Pr}\{\bigcup\limits_{j \neq 1} E_{2j}\} \leq e^{-n\gamma}.
 \end{equation}
 Combining all inequalities above, we can see that there must exist at least one $(n,M_n^{(1)},M_n^{(2)},\varepsilon_n^{(1)}, \varepsilon_n^{(2)})$ code satisfying~(\ref{Lamma_1_inequa}).
\end{proof}

\begin{lemma}\label{Error_lemma_2}
For all $n$, any $(n,M_n^{(1)},M_n^{(2)},\varepsilon_n^{(1)}, \varepsilon_n^{(2)})$ code satisfies that
\begin{equation}\label{leq}
\begin{array}{l}
\varepsilon_n^{(1)} \geq {\rm Pr}\{\frac{1}{n} \log \frac{P_{Y_1^n|X_1^n}({Y}_1^n | {X}_1^n)}{P_{{Y}_1^n}({Y}_1^n)} \leq \frac{1}{n} \log M_n^{(1)} - \gamma\} - e^{-n\gamma}\\
\varepsilon_n^{(2)} \geq {\rm Pr}\{\frac{1}{n} \log \frac{P_{Y_2^n|X_2^n}({Y}_2^n | {X}_2^n)}{P_{{Y}_2^n}({Y}_2^n)} \leq \frac{1}{n} \log M_n^{(2)} - \gamma\} - e^{-n\gamma},
\end{array}
\end{equation}
for every $\gamma > 0$, where $X_1^n~({\rm resp.,}\,X_2^n)$ places probability mass $1/M_n^{(1)}~({\rm resp.,}\,1/M_n^{(2)})$ on each codeword for encoder 1~(resp., encoder 2) and (\ref{P_independ}), (\ref{Py1x1}), (\ref{Py2x2}) hold.
\end{lemma}

\begin{proof}[Proof of Lemma~\ref{Error_lemma_2}]
The proof is similar to that of~\cite[Lemma~4]{Han93}. By using the relation
$$ \frac{P_{Y_1^n|X_1^n}({\bf y}_1 | {\bf x}_1)}{P_{{Y}_1^n}({\bf y}_1)} =  \frac{P_{X_1^n|Y_1^n}({\bf x}_1 | {\bf y}_1)}{P_{{X}_1^n}({\bf x}_1)}$$
and noticing $P_{{X}_1^n}({\bf x}_1) = \frac{1}{M_n^{(1)}}$, we can rewrite the first term on the right-hand side of the first inequality of (\ref{leq}) as
$${\rm Pr}\{P_{X_1^n|Y_1^n}(X_1^n|Y_1^n) \leq e^{-n\gamma}\}.$$
By setting
$$L_n = \{({\bf x}_1,{\bf y}_1)|P_{X_1^n|Y_1^n}({\bf x}_1 | {\bf y}_1) \leq  e^{-n\gamma}\},$$
the first inequality of (\ref{leq}) can be expressed as
\begin{equation}
{\rm Pr}\{L_n\} \leq \varepsilon_n^{(1)} +  e^{-n\gamma}.
\end{equation}
In order to prove this inequality, we set
$$\mathcal{A}_i = \{{\bf y}_1 \in \mathcal{Y}_1^n| P_{X_1^n|Y_1^n}({\bf x}_1(i) | {\bf y}_1) \leq  e^{-n\gamma}\}.$$
It can be seen that
\begin{equation}
\begin{array}{l}
{\rm Pr}\{L_n\} = \sum\limits_{i = 1}^{M_n^{(1)}} P_{X_1^nY_1^n}({\bf x}_1(i),\mathcal{A}_i)\\
= \sum\limits_{i = 1}^{M_n^{(1)}} P_{X_1^nY_1^n}({\bf x}_1(i),\mathcal{A}_i\bigcap \mathcal{B}_{1i}) + \sum\limits_{i = 1}^{M_n^{(1)}} P_{X_1^nY_1^n}({\bf x}_1(i),\mathcal{A}_i\bigcap \mathcal{B}_{1i}^c)\\
\leq \sum\limits_{i = 1}^{M_n^{(1)}} P_{X_1^nY_1^n}({\bf x}_1(i),\mathcal{A}_i\bigcap \mathcal{B}_{1i}) + \sum\limits_{i = 1}^{M_n^{(1)}} P_{X_1^nY_1^n}({\bf x}_1(i),\mathcal{B}_{1i}^c)\\
= \sum\limits_{i = 1}^{M_n^{(1)}}\sum\limits_{{\bf y}_1 \in \mathcal{A}_i\bigcap \mathcal{B}_{1i}} P_{X_1^nY_1^n}({\bf x}_1(i),{\bf y}_1) + \varepsilon_n^{(1)} \\
=\sum\limits_{i = 1}^{M_n^{(1)}}\sum\limits_{{\bf y}_1 \in \mathcal{A}_i\bigcap \mathcal{B}_{1i}} P_{X_1^n|Y_1^n}({\bf x}_1(i)|{\bf y}_1)P_{Y_1^n}({\bf y}_1) + \varepsilon_n^{(1)} \\
\stackrel{a}{\leq} e^{-n\gamma}\sum\limits_{i = 1}^{M_n^{(1)}}\sum\limits_{{\bf y}_1 \in \mathcal{B}_{1i}}P_{Y_1^n}({\bf y}_1)+ \varepsilon_n^{(1)} \\
= e^{-n\gamma}P_{Y_1^n}(\bigcup\limits_{i = 1}^{M_n^{(1)}}\mathcal{B}_{1i})+ \varepsilon_n^{(1)} \leq e^{-n\gamma} + \varepsilon_n^{(1)},
\end{array}
\end{equation}
where $\mathcal{B}_{1i}$ is the decoding region corresponding to codeword ${\bf x}_1(i)$ and $``a"$ follows from the definition of $\mathcal{A}_i$. Therefore, the first inequality of (\ref{leq}) is proved. Similarly, we can obtain the second inequality of (\ref{leq}).
\end{proof}

Now we prove Theorem~\ref{Capacity_Theorem}.
\begin{proof}[Proof of Theorem~\ref{Capacity_Theorem}]

1)~To prove that an arbitrary $(R_1,R_2)$ satisfying (\ref{Capacity_R1}) and (\ref{Capacity_R2}) is achievable, we define
$$M_n^{(1)} = e^{n(R_1-2\gamma)}\,\,\,{\rm and}\,\,\,M_n^{(2)} = e^{n(R_2-2\gamma)}$$
for an arbitrarily small constant $\gamma > 0$. Then, Lemma~\ref{Error_lemma_1} guarantees the existence of an $(n,M_n^{(1)},M_n^{(2)},\varepsilon_n^{(1)}, \varepsilon_n^{(2)})$ code satisfying
\begin{equation}\label{Direct_leq1}
\begin{array}{l}
\varepsilon_n^{(1)} + \varepsilon_n^{(2)} \leq {\rm Pr}\{\frac{1}{n} \log \frac{P_{Y_1^n|X_1^n}({Y}_1^n | {X}_1^n)}{P_{{Y}_1^n}({Y}_1^n)} \leq  R_1 - \gamma\}\\
+{\rm Pr}\{\frac{1}{n} \log \frac{P_{Y_2^n|X_2^n}({Y}_2^n | {X}_2^n)}{P_{{Y}_2^n}({Y}_2^n)} \leq  R_2 - \gamma\} + 2e^{-n\gamma}\\
 \leq {\rm Pr}\{\frac{1}{n} \log \frac{P_{Y_1^n|X_1^n}({Y}_1^n | {X}_1^n)}{P_{{Y}_1^n}({Y}_1^n)} \leq  \underline{I}({\bf X}_1;{\bf Y}_1)- \gamma\} \\
 + {\rm Pr}\{\frac{1}{n} \log \frac{P_{Y_2^n|X_2^n}({Y}_2^n | {X}_2^n)}{P_{{Y}_2^n}({Y}_2^n)} \leq  \underline{I}({\bf X}_2;{\bf Y}_2)- \gamma\} + 2e^{-n\gamma}.
\end{array}
\end{equation}
From the definition of the spectral inf-mutual information rate, we have
$$\lim_{n \rightarrow \infty} \varepsilon_n^{(1)} = 0 \,\,\,{\rm and}\,\,\,\lim_{n \rightarrow \infty} \varepsilon_n^{(2)} = 0.$$

2)~Suppose that a rate pair $(R_1,R_2)$ is achievable. Then, for any constant $\gamma > 0$, there exists an $(n,M_n^{(1)},M_n^{(2)},\varepsilon_n^{(1)}, \varepsilon_n^{(2)})$ code satisfying
\begin{equation}
\frac{\log M_n^{(1)}}{n} \geq R_1 - \gamma \,\,\,{\rm and}\,\,\, \frac{\log M_n^{(2)}}{n} \geq R_2 - \gamma
\end{equation}
for all sufficiently large $n$ and
$$\lim_{n \rightarrow \infty} \varepsilon_n^{(1)} = 0 \,\,\,{\rm and}\,\,\,\lim_{n \rightarrow \infty} \varepsilon_n^{(2)} = 0.$$
From Lemma~\ref{Error_lemma_2}, we get
\begin{equation}\label{geqR}
\begin{array}{l}
\varepsilon_n^{(1)} \geq {\rm Pr}\{\frac{1}{n} \log \frac{P_{Y_1^n|X_1^n}({Y}_1^n | {X}_1^n)}{P_{{Y}_1^n}({Y}_1^n)} \leq R_1 - 2\gamma\} - e^{-n\gamma}\\
\varepsilon_n^{(2)} \geq {\rm Pr}\{\frac{1}{n} \log \frac{P_{Y_2^n|X_2^n}({Y}_2^n | {X}_2^n)}{P_{{Y}_2^n}({Y}_2^n)} \leq R_2 - 2\gamma\} - e^{-n\gamma}
\end{array}.
\end{equation}
Taking the limits as $n\rightarrow \infty$ on both sides, we have
\begin{equation}\label{geqR1}
\begin{array}{l}
\lim\limits_{n\rightarrow \infty}{\rm Pr}\{\frac{1}{n} \log \frac{P_{Y_1^n|X_1^n}({Y}_1^n | {X}_1^n)}{P_{{Y}_1^n}({Y}_1^n)} \leq R_1 - 2\gamma\} = 0\\
\lim\limits_{n\rightarrow \infty}{\rm Pr}\{\frac{1}{n} \log \frac{P_{Y_2^n|X_2^n}({Y}_2^n | {X}_2^n)}{P_{{Y}_2^n}({Y}_2^n)} \leq R_2 - 2\gamma\} = 0
\end{array},
\end{equation}
implying by definition that $R_1 - 2\gamma \leq \underline{I}({\bf X}_1;{\bf Y}_1)$ and $R_2 - 2\gamma \leq \underline{I}({\bf X}_2;{\bf Y}_2)$. This completes the proof since $\gamma$ is arbitrary.
%
\end{proof}
\begin{figure}
  \includegraphics[width=8.5cm]{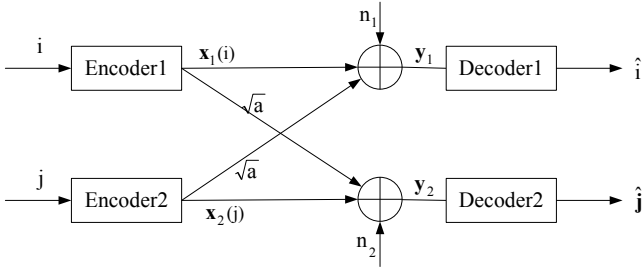}\\
  \caption{Symmetric Gaussian interference channel.}\label{GIFC}
\end{figure}
\section{Numerical Results}\label{GIFC_section}
We have obtained a formula of the capacity region for the¡¡
general IC, which shows that any pair of independent input
processes define a pair of achievable rates. Although it is
not computable in general, the formula provides us useful
insights into the interference channels, as illustrated by the following example.

The considered example has the model as shown in Fig.~\ref{GIFC}, where the channel inputs
${\bf x}_1(i)$ and ${\bf x}_2(j)$ are BPSK sequences with power constraints $P_1$ and $P_2$, respectively; the additive noise ${\bf n}_1$ and ${\bf n}_2$ are sequences of independent and identically distributed~(i.i.d.) Gaussian random variables of variance one; the channel outputs ${\bf y}_1$ and ${\bf y}_2$ are
\begin{eqnarray}
{\bf y}_1 &=& {\bf x}_1(i) + \sqrt{a}{\bf x}_2(j) + {\bf n}_1,\\
{\bf y}_2 &=& {\bf x}_2(j) + \sqrt{a}{\bf x}_1(i) + {\bf n}_2.
\end{eqnarray}
For any two arbitrary input processes ${\bf x}_1$ and ${\bf x}_2$, the defined pair of rates given in Theorem~\ref{Capacity_Theorem} are infeasible to compute.  Therefore, we assume that ${\bf x}_1$ and
${\bf x}_2$ are the outputs from two~(possibly different) generalized trellis encoders driven by independent uniformly distributed~(i.u.d.) input sequences, as proposed in~\cite{XiaoMa2011}.  In this case, both ${\bf x}_1$ and
${\bf x}_2$ are block-wise stationary, and~(hence)
\begin{eqnarray}
\underline{I}({\bf X}_1;{\bf Y}_1) &=& \lim_{n \rightarrow \infty}\frac{1}{n}I({X}_1^n;{Y}_1^n),\label{Spectrum_information_1}\\
\underline{I}({\bf X}_2;{\bf Y}_2) &=& \lim_{n \rightarrow \infty}\frac{1}{n}I({X}_2^n;{Y}_2^n).\label{Spectrum_information_2}
\end{eqnarray}
Since ${\bf x}_1, {\bf x}_2$ and ${\bf y}_1, {\bf y}_2$ can be viewed as the Markov chains and the hidden Markov chains, respectively, the right-hand sides of~(\ref{Spectrum_information_1}) and~(\ref{Spectrum_information_2}) can be estimated by performing the BCJR algorithm~\footnote{Only forward recursion is required.}~\cite{Arnold06}\cite{BCJR74}.

We consider two schemes, UnBPSK and CcBPSK, where UnBPSK means that ${\bf x}_1$~(resp.~${\bf x}_2$) is an i.u.d.~BPSK sequence and CcBPSK means that  ${\bf x}_1$~(resp.~${\bf x}_2$) is an output from the convolutional encoder with the generator matrix
\begin{equation}\label{Generator}
G(D) = [1 + D + D^2\,\,\,1+D^2].
\end{equation}
Fig.~\ref{FigTRIFFC} shows the trellis representation of the signal model when sender~1 uses CcBPSK and sender~2 uses UnBPSK.
Fig.~\ref{UncodedVSConv_figure} shows the numerical results. The point ``A" can be achieved by a coding scheme, in which
sender~1 uses a binary linear~(coset) code concatenated with the convolutional code and sender~2 uses a binary linear code, and the point ``B" can be achieved similarly; while the points on the line ``AB" can be achieved by time-sharing scheme. The point ``C" represents the limits when the two users use binary linear codes but regard the interference as an i.u.d. additive~(BPSK) noise. It can be seen that knowing the structure of the interference can be used to improve potentially the bandwidth-efficiency.
\begin{figure}
  \centering
  \includegraphics[width=8.5cm]{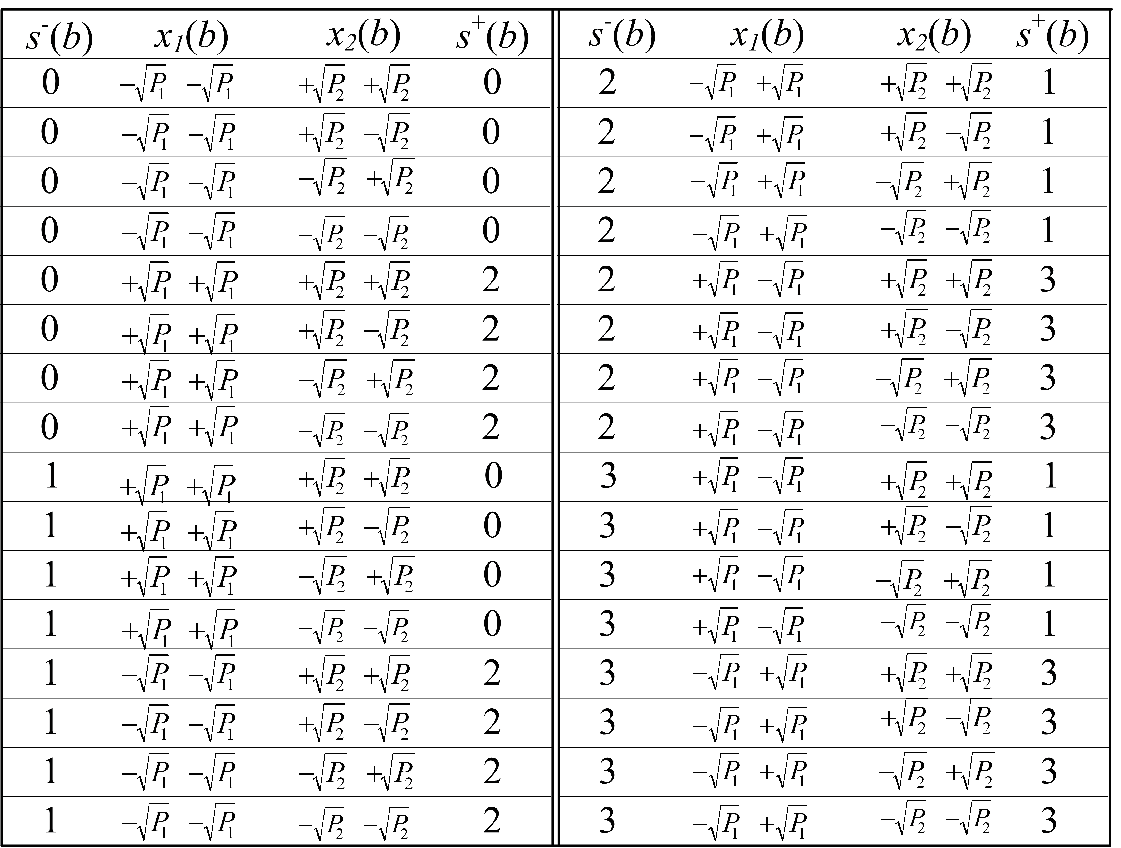}\\
  \caption{The trellis section of~(CcBPSK, UnBPSK) with 32 branches. For each branch $b$, $s^-(b)$ and $s^+(b)$ are the starting state and the ending state, respectively; while the associated labeling $x_1(b)$ and $x_2(b)$ are the transmitted signals at sender~1 and sender~2, respectively.}\label{FigTRIFFC}
\end{figure}


\begin{figure}
\centering
  \includegraphics[width=8.5cm]{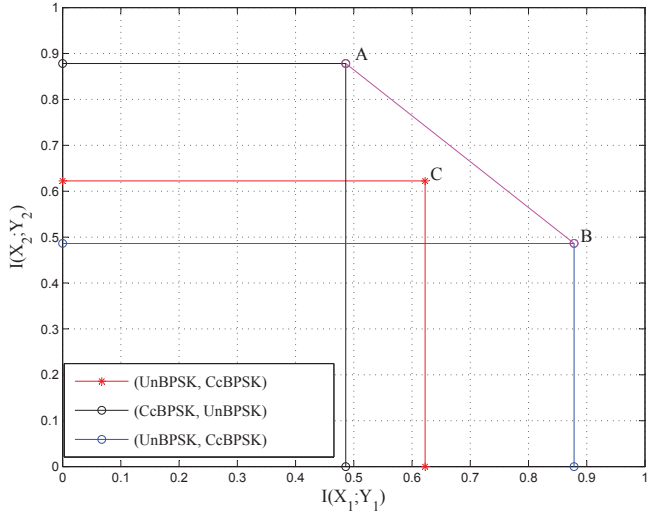}\\
  \caption{The evaluated achievable rate pairs of a specific GIFC, where $P_1=P_2=7.0~{\rm dB}$ and $a = 0.5$. UnBPSK means that ${\bf x}_1$~(resp.~${\bf x}_2$) is an i.u.d.~BPSK sequence and CcBPSK means that  ${\bf x}_1$~(resp.~${\bf x}_2$) is an output from the convolutional encoder with the generator matrix
$[1 + D + D^2\,\,\,1+D^2]$.}\label{UncodedVSConv_figure}
\end{figure}

\section{Conclusions}\label{Conclusions}
In this paper, we have proved that the capacity region of the two-user interference channel is the union of a family of rectangles, each of which is defined by a pair of spectral inf-mutual information rates associated with two independent input processes. For special cases, the defined pair of rates can be computed, which provide us useful insights into the interference channels.





%

\bibliographystyle{IEEEtran}
\bibliography{IEEEabrv,mybibfile}


\end{document}